\documentclass[a4paper,12pt,reqno]{amsart}

\usepackage{
amsfonts,
amsmath,
amsopn,
amssymb,
amsthm,
bbm,
bbold,
dsfont,
enumitem,
graphicx,
mathrsfs,
mathtools,
soul,
subfig,
verbatim,
xcolor,
xspace,
tikz,
mathabx,
}

\usetikzlibrary{decorations.markings,decorations.pathmorphing,arrows.meta}

\usepackage{geometry}
\usepackage[utf8]{inputenc}
% ...
\inputencoding{latin1}
% ...
\inputencoding{utf8}
% ...

\mathtoolsset{showonlyrefs}

\usepackage{dirtytalk}

\numberwithin{equation}{section}

\theoremstyle{plain} 
\newtheorem{thm}{Theorem}[section] 
\newtheorem{cor}[thm]{Corollary} 
 
\newtheorem{prop}[thm]{Proposition} 
\theoremstyle{definition}

\theoremstyle{definition} 

\theoremstyle{remark} 
\newtheorem{rem}[thm]{Remark}

\newcommand{\f}[2]{\frac{#1}{#2}}

\newcommand{\p}{\partial}

\DeclareMathOperator*{\Tr}{Tr_\pi}

\newcommand{\N}{\mathbb{N}}
\newcommand{\NN}{\mathscr{N}}
\newcommand{\R}{\mathbb{R}}
\newcommand{\C}{\mathbb{C}}
\newcommand{\D}{\mathcal{D}}
\newcommand{\T}{\mathcal{T}}

\newcommand{\F}{\mathcal{F}}

\newcommand{\x}{\mathbf{x}}
\newcommand{\y}{\mathbf{y}}

\renewcommand{\k}{\mathbf{k}}
\renewcommand{\p}{\mathbf{p}}
\newcommand{\mm}{\mathbf{m}}
\newcommand{\nn}{\mathbf{n}}

\newcommand{\E}{\mathcal{E}}
\newcommand{\B}{\mathcal{B}}

\renewcommand{\S}{\mathcal{S}}
\newcommand{\G}{\mathcal{G}}

\newcommand{\ve}{\varepsilon}

\newcommand{\al}{\alpha}

\newcommand{\Ga}{\Gammal}
\newcommand{\de}{\delta}
\newcommand{\beq}{\begin{equation}}
\newcommand{\eeq}{\end{equation}}
\newcommand{\lf}{\left}
\newcommand{\ri}{\right}

\newcommand{\n}{\noindent}
\newcommand{\vs}{\vspace{0.5cm}}

\newcommand{\la}{\lambda}
\renewcommand{\Ga}{\Gamma}
\newcommand{\ome}{\omega}

\renewcommand{\leq}{\leqslant}
\renewcommand{\geq}{\geqslant}

\title{On  Fermi's model for the  scattering of a slow neutron  from a bound proton
}

\author[D. Finco]{Domenico Finco}
\address[D. Finco]{Universit\`a Telematica Internazionale Uninettuno, Facolt\`a di Ingegneria, corso Vittorio Emanuele II, 00186, Roma, Italy.}
\email{d.finco@uninettunouniversity.net}

\author[R. Scandone]{Raffaele Scandone}
\address[R. Scandone]{Universit\`a degli Studi di Napoli ``Federico II'', Dipartimento di Matematica e Applicazioni ``R. Caccioppoli'', Via Cintia, Monte S. Angelo, 80126, Napoli, Italy}
\email{raffaele.scandone@unina.it}

\author[A. Teta]{Alessandro Teta}
\address[A. Teta]{Universit\`a degli Studi di Roma ``La Sapienza'', Dipartimento di Matematica ``G. Castelnuovo'', Piazzale Aldo Moro 5, 00185, Roma, Italy}
\email{teta@mat.uniroma1.it}

\date{}

\thanks{This work has been partially financed by European Union - Next Generation EU, Project MUR-PRIN 2022, ref. n. 2022CHELC7. The authors also acknowledge the support of the GNFM Gruppo Nazionale per la Fisica Matematica - INdAM 
and of the GNAMPA Gruppo Nazionale per l'Analisi Matematica, la Probabilit\`a e le loro Applicazioni - INdAM}

\begin{document}

\maketitle

\begin{abstract}
We consider a model Hamiltonian, introduced by Fermi in 1936, describing a two-particle system  made of a neutron and a harmonically bound proton, where the neutron-proton interaction has the form of a $\delta$-potential. For such Hamiltonian we prove the Limiting Absorption Principle and describe the stationary scattering theory. Finally, we derive Fermi's formula for the scattering cross-section valid in the Born approximation. 
\end{abstract}

\vs

\section{Introduction}

In the 1930s, Fermi and his research group in Rome were involved 
in fundamental experimental studies on radioactivity induced by slow neutrons. However, as is well known, a typical aspect of Fermi's scientific personality was to combine experimental activity with the development of theoretical models capable of clarifying and explaining the experimental results. Following this approach, in 1936 Fermi published an article  (\cite{F}) where he proposed a model for analysing the scattering of slow neutrons by protons. 
Taking into account that the proton-neutron interaction is extremely intense and its range is much shorter than the wavelength of the incident neutron, Fermi assumed that the interaction could be reasonably described by a Dirac delta function. Then he considered separately the three cases of a fixed proton, a free proton and a proton harmonically bound around an equilibrium position and he computed in each case the scattering cross-section in the Born approximation. In particular, in the case of a collision between the neutron and a harmonically bound proton he found the following   scattering cross-section (see also \cite{L} p. 644):

\begin{equation}\label{fermi}
\f{d \sigma_{\nn}}{d{\Omega}} = \f{4 a^2 }{ |\nn|!} \, \f{|\p|}{|\p_0|}  \left( \f{  (\p-\p_0)^2}{2 m \hbar \omega} \right)^{\!|\nn|} e^{- \f{ (\p-\p_0)^2}{ 2 m \hbar \omega}}  .
\end{equation}

\n
In \eqref{fermi}  the oscillator undergoes the transition from the ground state to the state labeled by $\nn=(n_1,n_2,n_3)$, with $|\nn|=n_1+n_2+n_3$, $a$ is a constant known as scattering length,    $\p_0, \p$ are initial and final momenta of the neutron, $m$ is the mass  and $\omega$ is the frequency of the oscillator. 
In the special case of elastic collision, formula \eqref{fermi}
reduces to:

\begin{equation}\label{fermi2}
\f{d \sigma_{\mathbf 0}}{ d{\Omega}} = 4 a^2\, e^{- \f{(\p-\p_0)^2}{ 2 m \hbar \omega}} \Big|_{|\p|=|\p_0|} .
\end{equation}

\n
From \eqref{fermi2}, one also compute the total cross-section 

\begin{equation}\label{fermi3}
\sigma_{\mathbf 0}  = 4 \pi a^2\,  \f{\hbar \omega}{E} \left( 1 - e^{- \f{4 E}{\hbar \omega} } \right),  \end{equation}

\n
where $E$ is the energy of the neutron. The above formulas show the explicit dependence of the collision process on the frequency of the oscillator, and therefore on the motion of the proton. 
It is important to emphasize once again that Fermi's approach was perturbative and, due to the singular nature of the interaction, could not go beyond Born's approximation. 

\n
In the same years, Bethe and Peierls (\cite{BP}) developed a complete theory, i.e. valid at any perturbative order, for the case of a fixed proton and later Berezin and Faddeev (\cite{BF}) provided a rigorous mathematical construction of the Hamiltonian (see also \cite{Al} for a comprehensive treatment of the whole subject). 

\n
Here we are interested in the case of a harmonically bound proton. The corresponding Hamiltonian can be formally written as:

\begin{equation}\label{ht}
\widetilde{H} = -\Delta_{\x} - \Delta_{\y} + \f{1}{4} \, \omega^2 |\y|^2 - \f{3}{2} \,\omega + \delta (\x-\y),
\end{equation}

\n
where $\x, \y$ denote the coordinates of the neutron and the proton respectively. For simplicity, we have fixed $\hbar =1$, the masses of both particles equal to $1/2$ and the ground state energy of the oscillator equal to zero.  
Notice that in \eqref{ht} the unperturbed Hamiltonian is 

\begin{equation}\label{frh}
H_0= -\Delta_{\x} - \Delta_{\y} + \f{1}{4} \, \omega^2 |\y|^2 - \f{3}{2} \,\omega,
\end{equation}

\n
and the interaction term is non trivial only on the coincidence hyperplane 

\begin{equation}\label{hypl}
\pi= \big\{ (\x,\y)\in \R^6 \,|\, \x=\y \big\}.
\end{equation}

\n
A rigorous counterpart $H, \D(H)$ of \eqref{ht}  as a well defined self-adjoint and bounded from below operator in $L^2(\R^6)$ was constructed  in \cite{FDT} (see also \cite{CDF}) using a renormalisation technique and the theory of quadratic forms. Roughly speaking, the Hamiltonian  $H, \D(H)$ is characterised by the following two conditions 

%
%
%

\vs

\begin{itemize}
   \item[(C-1)] $\;H\psi=H_0\psi$ for any $\psi \in \D(H)$ such that $\text{Tr}_\pi \psi =0$,  
where  $\text{Tr}_\pi  $ denotes the trace on $\pi$; 

\vspace{0.2cm}

    \item[(C-2)] $\; \psi \in \D(H)$ satisfies the boundary condition on $\pi$
   \begin{equation}\label{BC}
\psi(\x,\y)= \f{\xi(\y)}{|\x-\y|} + \alpha\, \xi(\y) + o(1) \;\;\;\;\;\; \text{for}\;\; \x \to \y,
    \end{equation}
  where
    \begin{equation}
    \xi(\y):=  \lim_{\x \to \y} |\x-\y| \psi(\x,\y) ,
    \end{equation}
   and $\alpha\-\in\-\R$ is the inverse of the two-body scattering length.
\end{itemize}

\vs

\n
Condition (C-1)  expresses that fact that the interaction takes place only when the positions of the neutron and the proton coincide,   
while (C-2) is the boundary condition imposed on $\pi$ and it is the direct generalisation of the Bethe and Peierls boundary condition for the case of fixed proton. 

\n
The aim of this work is to show the validity of the Limiting Absorption Principle (LAP) and to develop the stationary scattering theory for $H, \D(H)$. As a final result, we also derive Fermi's formula \eqref{fermi}  in the Born approximation. 

\n
In the next Section we will recall the rigorous definition and the main properties of $H, \D(H)$ and then we will give the precise formulation of our main results.

\vs
\n
For the convenience of the reader, we collect here some notation used throughout the paper. 
We denote the integer numbers with zero included by $\N_0 = \N \cup \{0\}$. We adopt a unitary definition of the
Fourier transform in $L^2(\R^d)$
$$
\hat{f} (\k) = \f{1}{(2 \pi)^{d/2}} \int_{\R^d}\! \! d \x \, e^{- i \k \cdot \x} \, f(\x).
$$
Weighted Sobolev spaces are defined in the following way:
\begin{align}
&L^2_s(\R^d) = \{ f: \R^d\to \C\, | \langle x \rangle^s f \in L^2(\R^d)     \} \qquad \|f \|_{L^2_s (\R^d)}=   \| \langle x \rangle^s f \|_{L^2 (\R^d)}       \;\; s\in\R  ,   \nonumber\\
& H^2_s(\R^d) = \{ f: \R^d\to \C\, | f,\, \Delta f \in L^2_s(\R^d)     \}       \qquad \|f \|_{H^2_s (\R^d)}=    \|f \|_{L^2_s (\R^d)}+ \|\Delta f \|_{L^2_s (\R^d)}     .      \nonumber
\end{align}
Given two Hilbert spaces $\mathcal K_1$ and $\mathcal K_2$, we denote by $\B(\mathcal K_1, \mathcal K_2)$ the space of bounded linear operators 
from $\mathcal K_1$ to $\mathcal K_2$ equipped with the uniform topology, while $\B(\mathcal K_1)$ stands for $\B(\mathcal K_1, \mathcal K_1)$.
We also denote the closed ideal of compact operators in $\B(\mathcal K_1)$ by $\B_\infty (\mathcal K_1)$.

\n
For a free particle in $L^2(\R^3)$ we have
\begin{align} \label{pasta1}
&h_0=-\Delta_{\x} \,, \;\;\;\;\;\; \D(h_0)= H^2(\R^3) \,, \;\;\;\;\;\; r_0(z)=(h_0 -z)^{-1} \,, \;\;\;\;\;\; z \in \C \setminus \overline{\R^+}.
\end{align}

\n
For an harmonic oscillator in $L^2(\R^3)$ we have
\begin{align} \label{pasta2}
&h_{\omega}=-\Delta_{\y} + \f{1}{4} \, \omega^2 |\y|^2 - \f{3}{2} \,\omega\,, \;\;\;\;\;\; \D(h_{\omega})= H^2(\R^3)\cap L^2_2 (\R^3) \,, \;\;\;\;\;\; r_{\omega}(z)=(h_{\omega} -z)^{-1} \,, \;\;\;\;\;\; 
z \in \C \setminus \N_0.
\end{align}
Notice that the characterization of $\D(h_{\omega})$ follows from the separation property, see \cite{DG, Eva,Eve}.
For  $ \nn =(n_1,n_2,n_3) \in \N_0^3\,,$ $\; \phi_{\nn}$ is a normalized eigenfunction of  $  h_{\omega}$ corresponding to the eigenvalue 
$\omega |\nn|$, where $|\nn|= (n_1+n_2+n_3) \,.$ The degeneracy of the eigenvalue $\,\omega |\nn|\,$ is $\, \f{1}{2} (|\nn|+1)(|\nn| +2)$.

\vs

\section{Model Hamiltonian and main results}

\noindent
We first recall the definition and  some properties of our model Hamiltonian (see  \cite{FDT}, \cite{CDF}). Let us denote the free resolvent
\begin{equation}
R_0(z)=(H_0 -z)^{-1} \,, \;\;\;\;\;\; z \in \C \setminus \overline{\R^+}
\end{equation}

\n
and the corresponding integral kernel
\begin{align}
&G^z (\x,\y;\x',\y')= \!\sum_{\nn \in \N^3_0} r_0(z - \omega |\nn|) (\x,\x') \,\phi_{\nn}(\y) \phi_{\nn} (\y')
\!=\! \sum_{\nn \in \N^3_0} \int_{\R^3} \!\!d\k\, \f{1}{|\k|^2 \!+\! \omega |\nn|\! -\!z} \f{e^{i \k \cdot (\x-\x')}}{(2\pi)^{3}} \,\phi_{\nn}(\y) \phi_{\nn} (\y').
\end{align}
%
Let us also introduce the quadratic form in $L^2(\R^3)$ for any $\lambda >0$

\begin{align}
&\D (\Phi^{\lambda})= H^{1/2} (\R^3) \cap L^{2}_{1/2} (\R^3)\\
&\Phi^\la (\xi ) = \f12 \int_{\R^6} G^{-\lambda} (\y , \y ; \y' , \y' ) | \xi (\y) - \xi( \y')|^2  \, d\y \, d\y' + \int_{\R^3} a^{\lambda} (\y) |\xi (\y) |^2 d\y
\end{align}

\n
where 
$a^{\lambda}$ is a spherically symmetric and increasing function such that $ \lim_{\la \to +\infty} a^{\lambda} (0) =+\infty$, see \cite{CDF} for the explicit expression. In \cite{FDT} it is proved that there exists $\lambda_0>0$ such that for any $\lambda >\lambda_0$ the quadratic form $\Phi^{\lambda}, \D(\Phi^{\lambda})$ is closed and 
\begin{equation}\label{coerc}
\Phi^{\lambda} (\xi ) \geq a^{\lambda}(0)\,\| \xi \|^2_{L^2(\R^3) }
\end{equation}
 Let us denote by $\Gamma(-\lambda)$, $\D(\Gamma(-\lambda))$ the unique self-adjoint and positive operator associated to $\Phi^{\lambda}, D(\Phi^{\lambda})$ for  $\lambda >\lambda_0$. Let us also define the \say{potential} produced by the \say{charge distribution} $\xi $
\begin{equation}
\mathcal G(z) \xi (\x,\y) = \int_{\R^3} \!\!  \, G^z (\x,\y;\y',\y') \xi(\y')\, d\y'
\end{equation}
and notice that $\mathcal G(-\lambda) \in \mathcal B(L^2(\R^3), L^2(\R^6))$. Then for any $\alpha \in \R $ we define our model Hamiltonian in $L^2(\R^6)$ as 

\begin{align} \label{doha}
&\D (H) = \big\{ \psi \in L^2 (\R^6)\, | \,\psi =\phi^\la+ \mathcal G (-\lambda)  \xi , \, \phi^\la\in \D(H_0) ,\, \xi \in \D (\Gamma (-\lambda)), \, \nonumber\\
& \;\;\;\;\;\;\;\;\;\;  \;\;\;\;\;\;\;\;\;\;\;\;\;\;\;\;\;\;\;\;\;\;\;\;\;\;\;\;\;\;\;\;\;\;\;\;\;\;\;\; \;\;\;\;\;\;\;\;\;\;\;\;\;\;\;\;\;\;\;\;\;\;\;\;\;\;\;\;\;\;\;\;\;\;\;\;\;\;\;\;\big(\Gamma (-\lambda) +\al \big) \xi = \text{Tr}_\pi \phi^\la  \big\}, \\
&(H+\la) \psi = (H_0 +\la) \phi^\la \label{azha} .\\
\end{align}
The definition is independent of the particular choice of $\lambda $. It turns out that $H, \D(H)$ is self-adjoint and bounded from below. 
Taking $\la>\la_0$,  $\lambda_0$ sufficiently large, we have that  $\, \Gamma (-\lambda)+\al \,$ is invertible 
and $\big(\Gamma (-\lambda)+\al \big)^{-1} \!\in \B (L^2 (\R^3) )$.  
Then we can represent the resolvent 
\beq
R(z)=(H-z)^{-1}\,, \;\;\;\;\;\;\;\;\;\; z \in \C \setminus \sigma(H)
\eeq
 for $z=-\lambda$, $\lambda>\lambda_0$, as 

\begin{equation}\label{tama}
R(-\la)= R_0 (-\la) + \G (-\la) (\Gamma (-\la) +\al )^{-1} \G^\ast ( -\la)\,. 
\end{equation}

\n
In the following $\la$ will be a fixed number such that $\la>\la_0$. Formula \eqref{tama} shows that the free Hamiltonian is recovered in the limit $\alpha \to + \infty$, i.e., when the two-body scattering length $a$ goes to zero. The relation between $a$ and $\al$ will be discussed in Section 5.

\n
One can check that $H, \D (H)$ defined in \eqref{doha}, \eqref{azha} satisfies the conditions C-1, C-2 written in the Introduction. In particular, for $\, |\x-\y| \to 0\,$  one has the asymptotic expansion

\[
 \G (-\la) \xi   (\x,\y)= \frac{\xi(\y) }{|\x -\y|} -  \Ga(-\la)\xi (\y)+\text{o}(1)
\]
and therefore, for any $\psi \in \D(H)$ one finds 
\begin{align}
&\psi(\x,\y)= \phi^{\lambda}(\x,\y) + \G(-\lambda) \xi (\x,\y) =  \f{\xi(\y)}{|\x-\y|}+\text{Tr}_\pi \phi^\la (\y) - \Gamma(-\lambda) \xi (\y ) + \text{o}(1)\,.
\end{align}
for $|\x -\y|\to 0$. Using the last equation   in \eqref{doha}, we find that $\psi$ satisfies the boundary condition \eqref{BC}. In other words, the equation $\big(\Gamma (-\lambda) +\al \big) \xi = \text{Tr}_\pi \phi^\la$ in \eqref{doha} is an equivalent way of writing  the boundary condition \eqref{BC}.  
We also recall that in \cite{CDF} it is proved that the essential spectrum of $H, \D(H)$ is $[0,+\infty)$ and the wave operators exist and are complete. 

\n
Let us now formulate our main results. 
As a first preliminary step, we show that representation \eqref{tama} extends to the complex plane.  

\begin{prop}\label{meroR}
The resolvent $R (z)$ is a meromorphic operator valued function on $\C \setminus \overline{\R^+}$ with values in  $\B(L^2 (\R^6) )$ and 
\beq \label{continuation}
R(z)= R_0 (z)+ \G (z) (\Gamma (z) +\al )^{-1} \G^\ast ( \overline{z} ) \qquad z\in \C\setminus (\R^+ \cup \sigma_p (H))
\eeq

\n
where $\mathcal G(z) \in \B(L^2(\R^3), L^2(\R^6))$,   
\beq \label{holomo}
\Gamma (z) = \Gamma (-\la) - (\la+z) \G^\ast(-\la ) \G(z) \qquad z\in \C \setminus \R^+
\eeq

\n
and $ (\Gamma (z) +\al )^{-1} \in \mathcal B(L^2(\R^3))\,$ for $\,z\in \C\setminus (\R^+ \cup \sigma_p (H))$.
\end{prop}

\vs

\n
Starting from the representation \eqref{continuation}, we prove the LAP for $H, \D(H)$.

\begin{thm} \label{LAPmain}
There exists a discrete set $\NN \subset [0, \infty)$ such that for $\mu \in [0,\infty)\setminus \NN$ and $\mu/\ome  \notin \N_0$ the limits
\beq\label{bvR}
 \lim_{\ve \downarrow 0} R(\mu \pm i \ve) := R^\pm (\mu) = R_0^{\pm}(\mu)+ \mathcal G^{\pm}(\mu) \lf( \Gamma^{\pm}(\mu) + \alpha \ri)^{-1} \mathcal G^{* \mp}(\mu) 
\eeq
exist in $\B (L^2_{s} (\R^3) , L^2_{-s} (\R^3)) \otimes \B (L^2(\R^3))$ for $s>1/2$ and they are continuous in $\mu$ in the uniform topology. The same statement holds for $s>1$ if $\mu/\ome  \in \N_0$. 
Moreover, if we denote by $\E$ the set of  positive eigenvalues of $H$ of finite multiplicity, we have $\NN \cap (\R^+\setminus \N) \subset \E$.
\end{thm}

\vs

\n
The proofs of Proposition \ref{meroR} and Theorem \ref{LAPmain} are given in Section 3, where also the mapping properties of each operator in the r.h.s. of \eqref{bvR} are characterized.  As a corollary we also obtain that $\;\sigma_{ac}(H) = [0,\infty)$, $\;\sigma_{sc}(H)=\emptyset\;$   and the wave operators exist and are complete. 

\n
The second main result is the eigenfunction expansion theorem for $H, \D(H)$ and the corresponding representation for the wave operators $W_{\pm}$. 
More precisely, let us define for $f\in L^2(\R^6)$
\[
 \lf( \F_0 f\ri) (\k, \nn) = \int_{\R^3} \overline{ \Phi_0 (\k, \nn, \x,\y) } f(\x,\y) \, d\x\, d\y    \qquad \qquad  \Phi_0 (\k, \nn, \x,\y) = \frac{1}{(2 \pi)^{3/2} } e^{ i \k\cdot \x} \phi_\nn (\y).
\]
The map $\F_0:  L^2(\R^6) \to  L^2(\R^3)\otimes \ell^2 (\N_0^3)$ is a unitary operator and it diagonalizes $H_0$, that is for $ f\in \D (H_0)$ we have
\[
(\F_0 \, H_0 f) (\k, \nn) = \lf(|\k|^2 + \ome |\nn| \ri) (\F_0 f) (\k, \nn).
\]
We also define the generalized eigenfunctions $\Phi_\pm (\k, \nn)$, $  {|\k|^2}+ \ome |\nn|\notin \NN $ of the Hamiltonian $H$ by
\beq 
\begin{cases} \label{borbotto}
\Phi_\pm (\k, \nn) = \Phi_0 (\k, \nn) + \G^\pm ( |\k|^2 + \ome |\nn| ) \xi_\pm (\k, \nn) \\
( \Ga^\pm ( |\k|^2 + \ome |\nn| ) + \al ) \xi_\pm (\k, \nn)  = \Tr \Phi_0 (\k, \nn)
\end{cases}
\eeq

\n
and  the generalized Fourier transform, for $f\in \S (\R^6) $
\beq
\F_\pm f (\k, \nn) = \int_{\R^6} \overline{ \Phi_\pm (\k, \nn, \x,\y) } f(\x,\y) \, d\x\, d\y.
\eeq

\n
Then we have

\begin{thm}\label{eigexp}
The generalized Fourier transform  extends  to a bounded map $\F_\pm : L^2(\R^6) \mapsto L^2(\R^3)\otimes \N_0^3$. 
Moreover 
\beq \label{call1}
W_\pm = \F_\pm^\ast \F_0,
\eeq
and for any Borel function $f:\R \to \R$
\beq \label{call2}
\F_\pm^\ast  \F_\pm = P_{ac}(H) \qquad  \F_\pm  \F_\pm^\ast=I \qquad f(H)\, P_{\text{ac}} (H) =\F_\pm^\ast M_{ f} \F_\pm,
\eeq
where $M_{f}$ is the multiplication operator by $f({ |\k|^2} + \ome |\nn |)$ .
\end{thm}

\vs

\n
The proof of Theorem \ref{eigexp} is given in Section 4. 

\n
As a final result, in Section 5  we show that Fermi's formula \eqref{fermi} for the cross-section can be recovered from the scattering matrix $S=W_{+}^* W_{-}$ in the Born approximation, i.e., for $\alpha$ large.


\vs

\section{Limiting Absorption Principle for $H$}
\n
In this section we prove the representation \eqref{continuation} for the resolvent $R(z)$ and the limiting absorption principle for $H$.  
We first prove Proposition \ref{meroR}. 

\begin{proof}

Let us observe that $R_0 (z)$ is an analytical operator valued function on $\C \setminus \overline{\R^+}$ with values in $\B(L^2(\R^6), \D(H_0))$. 
Due to \eqref{pasta1} and  \eqref{pasta2}  we have $\D (H_0)\subset  H^2(\R^6)$.
Hence $R_0 (z)$ is an analytical function on $\C \setminus \overline{\R^+}$ with values in  $\B(L^2 (\R^6) , H^2 (\R^6))$.
In particular, this implies that  $\G^\ast (z)=\Tr R_0(z)$ is analytical
on $\C \setminus \overline{\R^+}$ with values is in $\B(L^2(\R^6), L^2(\R^3))$ by trace theorems. By duality $\G (z) $ takes values in $\B(L^2(\R^3), L^2(\R^6))$. 

\n
Let us discuss the analytical continuation $\Ga(z)$ of $\Ga(-\la)$. 
We notice that the resolvent identity implies that for $\la,\mu>0$ 
\[
 \G (-\la) \xi = \G (-\mu) \xi +(\mu-\la)  R_0 (-\la)  \G (-\mu) \xi .
\]
Since $ \G (-\la) \xi \in L^2(\R^3) $ for $\xi \in L^2 (\R^3)$, the second term on the r.h.s. belongs to $H^2(\R^6)$  and therefore can be traced on $\pi$. 
This implies the following identity
\beq \label{ufficio}
\Gamma (-\mu) =\Gamma(-\la) - (\la-\mu) \G^\ast(-\la) \G(-\mu).
\eeq
Motivated by \eqref{ufficio}, 
we define the analytical continuation $\Ga(z)$ as in \eqref{holomo}. Moreover, for notational convenience we set 
\beq \label{Kappa}
K(z) := \G^\ast(-\la ) \G(z) \qquad z\in \C \setminus \R^+.
\eeq
We observe that $K(z)$ is an analytic operator valued function with values in $\B_\infty(L^2(\R^3))$ due to the above analiticity properties of $\G (z)$ and that $ \G^\ast(-\la )$ is a compact operator from $L^2(\R^6)$ to  $L^2(\R^3)$ (see \cite{CDF}). 
Using \eqref{holomo}, we can  cast $\Gamma (z) +\al$ in the following  form
\begin{align*}
\Gamma (z) +\al &= \lf(\Gamma(-\la) + \al \ri)^{1/2} \lf[ I - (\la+z)  (\Gamma(-\la) + \al )^{-1/2} K(z)  (\Gamma(-\la) + \al )^{-1/2} \ri] \lf( \Gamma(-\la) + \al \ri)^{1/2} \\
& =  \lf(\Gamma(-\la) + \al \ri)^{1/2} \lf(I -M(z)\ri) \lf( \Gamma(-\la) + \al \ri)^{1/2}
\end{align*}
where we have  introduced 
\[
M(z) :=  (\la+z)  (\Gamma(-\la) + \al )^{-1/2} K(z)  (\Gamma(-\la) + \al )^{-1/2} .
\]
We know that $(\Gamma(-\la) + \al )^{-1/2} \in \B(L^2(\R^3) )$ for $\la >\la_0$. 
Therefore $M(z)$ for $z\in \C \setminus \R^+$ is an analytic function on $\C \setminus \R^+$ with values in $\B_\infty(L^2(\R^3))$
such that $M(-\la)=0$. Hence, by analytical Fredholm theorem, $(I -M(z))^{-1}$ is meromorphic on $\C \setminus \R^+$ with values in  $\B (L^2(\R^3))$ and the same statement holds true for 
\beq \label{waraku}
(\Gamma(z) + \al )^{-1} =  (\Gamma(-\la) + \al )^{-1/2}(I -M (z))^{-1} (\Gamma(-\la) + \al )^{-1/2}.
\eeq
We notice that the poles of \eqref{waraku} correspond to the eigenvalues of the self-adjoint operator $H$.
The representation \eqref{continuation} follows from \eqref{tama} and the analyticity properties discussed above and therefore the proof of Proposition \ref{meroR} is complete.
\end{proof}

\vs

\n
In the rest of this Section we discuss the more delicate problem of the boundary limits of the resolvent $R(z)$ on the real line. For the convenience of the reader we  recall some well known results on the LAP for $h_0$  (see e.g. \cite{A,Y}).
\begin{prop} \label{LAP0}
Let $r_0(z)=(h_0-z)^{-1}$. We consider $r_0(z)$ as an operator valued analytical function on $\C \setminus \overline{\R^+}$ with values in $\B(L^2_s (\R^3), H^2_{-s} (\R^3))$ for $s>1/2$. Moreover if $\mu>0$ then the following limits exists in the  uniform 
topology of $\B(L^2_s (\R^3), H^2_{-s}(\R^3)) $ 
\[
\lim_{ \substack{ z\to \mu \\ \pm\text{Im}\,z>0}} r_0 (z) =r_0^\pm (\mu),
\]
and the limits are continuous functions of $\mu$. The same statement holds on  $\C \setminus {\R^+}$ for $\mu=0$ if $s>1$. The following identity holds for $f\in L^2_s (\R^3) $, $s>1/2$, $\mu>0$
\beq \label{agmontrace}
\text{Im} \langle r_0^\pm (\mu) f,f \rangle= \pm \f{\pi}{2\sqrt{\mu} }\int_{|\k |=\sqrt{\mu} } |\hat f (\k)|^2 d\sigma.
\eeq
\end{prop}

\vs

\n
The analysis is divided in several steps. We first discuss   the boundary limits of $R_0(z)$ and   $\G(z)$. Then we study the boundary limits of $(\Gamma (z) +\al )^{-1}$ and we show the relation between  their singularities and the positive eigenvalues of the Hamiltonian. 

\n
The main idea for the whole analysis is the decomposition of the relevant operators into a high energy part and a low energy part.
Using a representation by series, a crucial simplification is that the low energy part is always made by a finite number of terms. On the other hand the high energy part is related to the resolvent of a truncated unperturbed Hamiltonian. Taking the boundary limits of the high part is trivial
since we are considering  points in the resolvent set of the truncated Hamiltonian.

\n
In the following proposition we prove the LAP for the free resolvent $R_0(z)$.


\begin{prop} \label{LAPR0}
For $\mu>0$, $\mu \notin \N_0$, the following limits exists in the  uniform  topology of  $\B(L^2_s (\R^3)\otimes L^2 (\R^3) , H^2_{-s} (\R^3)\otimes H^2 (\R^3))$,  for $s>1/2$,  
\[
 \lim_{\ve \downarrow 0} R_0 (\mu \pm i \ve ) =R_0^\pm (\mu),
\]
and the limits are continuous functions of $\mu$ in the uniform topology. The same statement holds for $s>1$ if $\mu/\ome\in \N_0$.
\end{prop}
\begin{proof}
 As discussed in Proposition \ref{meroR}, the unperturbed resolvent $R_0 (z)$ is an analytical operator valued function on $\C \setminus \overline{\R^+}$ with values in  $\B(L^2 (\R^6) , H^2 (\R^6))$  and
a fortiori with values  in  $\B(L^2_s (\R^3)\otimes L^2 (\R^3) , H^2_{-s} (\R^3)\otimes H^2 (\R^3))$, $s>1/2$. 
We start from the following representation
\[
R_0(z)=  \sum_{\nn \in \N_0^3 } r_0 (z-\ome |\nn|) \otimes P_{\nn}
\]
where $P_{\nn}$ is the orthogonal projector on $\phi_{\nn}$. 
Let us fix $\mu>0$ and define $n_0 = \lfloor \mu/\omega \rfloor$. Then we decompose $R_0 (\mu +i \ve )$ into a low-energy term $ R_0^<(\mu+i\ve)$ and a high-energy term $ R_0^> (\mu+i\ve)$ defined as follows
\begin{align}
R_0 (\mu+i\ve ) & = \sum_{ |\nn| \leq n_0} r_0(\mu+i\ve -\ome |\nn|) \otimes P_\nn +  \sum_{ |\nn| \geq n_0+1} r_0(\mu+i\ve -\ome |\nn|) \otimes P_\nn \\
&=:  R_0^<(\mu+i\ve ) +  R_0^> (\mu+i\ve ).
\end{align}
Notice that $ R_0^<(\mu+i\ve )$ contains a finite number of terms and that $P_{\nn}$ maps $L^2 (\R^3)$ into $ \S (\R^3)$ and a fortiori into $H^2 (\R^3)$. 
Then, by Proposition \ref{LAP0}, the following limit exists in $\B(L^2_s (\R^3)\otimes L^2 (\R^3) , H^2_{-s} (\R^3)\otimes H^2 (\R^3))$ for $s>1/2$ if $\mu/\ome\notin \N_0$:
\[
 \lim_{\ve \downarrow 0}  R_0^<(\mu+i\ve ) =  \sum_{ |\nn| \leq n_0} r_0^+ (\mu -\ome |\nn|) \otimes P_{\nn}=:  R_0^{<,+}(\mu) ,
\]
and $R_0^{<,+}(\mu)$ is continuous in $\mu$ in such space. If $\mu \in \N_0$, then we require $s>1$ since the terms in $R_0^<(\mu )$ with $\ome |\nn|=\mu$ contains $r_0(0)$.

\n
For the high-energy part, we use a different argument.
Let $\chi_{n_0}$ be the characteristic function of $[ n_0+1, \infty)$ and $P^>_{n_0} = \chi_{n_0}(h_\ome)$. We introduce a truncated Hamiltonian $h_{\ome, n_0}$ defined by
\[
h_{\ome, n_0} :=  P^>_{n_0} h_{\ome} P^>_{n_0}= \sum_{ |\nn|\geq n_0+1}\ome |\nn | P_{\nn} = h_\ome -  \sum_{ |\nn|\leq n_0} \ome |\nn | P_{\nn},
\]
and we denote
 $$H_{0, n_0} =h_0 + h_{\ome, n_0} \,.$$ 
The resolvent $ ( H_{0, n_0} - z)^{-1}$ is analytical in $\C\setminus [n_0+1, \infty) $, since 
the spectrum of $ H_{0, n_0}$ is included in $ [n_0+1, \infty)$. Moreover $\D(h_{\ome, n_0})= \D(h_{\ome})$, since the two Hamiltonians differ by an $L^2$-bounded term. 
Therefore $\D ( H_{0, n_0} ) \subset  H^2(\R^6)$. Notice that the convergence of  $ R_0^> (\mu+i\ve )= (  H_{0, n_0} - (\mu+i\ve) )^{-1}$ to $R_0^> (\mu)$ is trivial. Indeed, $\mu$ belongs to the resolvent set of $ H_{0, n_0}$ and then $R_0^> (\mu)$ is analytical in a neighborhood of $\mu$.

\n
The same arguments can be repeated for $R_0 (\mu-i\ve ) $ with trivial modifications. We also notice that the limit for the high-energy parts is the same,
that is it does not depend on $\pm$. 
\end{proof}

\vs

\n
The next step is the characterization of the boundary limits of $\G(z)$. 

\begin{prop}\label{LAPpot}

For $\mu>0$, $\mu / \ome \notin \N_0$, the following limits exists in the  uniform  topology of $\B(L^2(\R^3) , L^2_{-s} (\R^3)\otimes L^2 (\R^3))$  
\[
 \lim_{\ve \downarrow 0} \G (\mu \pm i \ve ) =\G^\pm (\mu),
\]
and the limits are continuous functions of $\mu$. The same statement holds for $s>1$ if $\mu/ \ome \in \N_0$.
\end{prop}
\begin{proof}

We use  similar strategy as before, with a decomposition in an high-energy part  $\G^>(\mu+i\ve)$ and low-energy part $\G^<(\mu+i\ve)$ defined as follows:

\begin{align}
\G^<(\mu+i\ve ; \x, \y, \y')& = \sum_{ |\nn| \leq n_0} \frac{e^{i\sqrt{\mu - \ome |\nn| +i\ve}|\x-\y'|}}{4\pi |\x-\y'| } \phi_\nn (\y) \phi_\nn (\y')  \\
\G^> (\mu+i\ve ; \x, \y, \y')&=      \sum_{ |\nn| > n_0} \frac{e^{i\sqrt{\mu - \ome |\nn|+i\ve}|\x-\y'|}}{4\pi |\x-\y'| } \phi_\nn (\y) \phi_\nn (\y') \\
\G^{+,<} (\mu ; \x, \y, \y')& = \sum_{ |\nn| \leq n_0} \frac{e^{i\sqrt{\mu - \ome |\nn| +i\ve}|\x-\y'|}}{4\pi |\x-\y'| } \phi_\nn (\y) \phi_\nn (\y')  \\
\G^> (\mu ; \x, \y, \y')&=      \sum_{ |\nn| > n_0} \frac{e^{-\sqrt{\ome |\nn|-\mu}|\x-\y'|}}{4\pi |\x-\y'| } \phi_\nn (\y) \phi_\nn (\y').
\end{align}
For the low-energy part, we have a finite number of terms. Then it is sufficient to notice that for $\forall s \geq 0$ and $\forall \nn\in\N_0$, $\langle x\rangle^s \phi_\nn \in L^\infty (\R^3)$, and the convergence claim follows by Proposition \ref{LAP0}.
Namely we have
\begin{align}
& \| \G^<(\mu+i \ve) f  - \G^{+,<} (\mu) f \|_{ L^2_{-s} (\R^3)\otimes L^2 (\R^3) } 
 \leq c  \sum_{ |\nn| \leq n_0} \| r_0(\mu+i\ve  )\phi_\nn f  - r_0 (\mu)\phi_\nn f \|_{L^2_{-s} (\R^3) }\\
&\leq c  \sum_{ |\nn| \leq n_0} \| r_0(\mu+i\ve  ) - r_0 (\mu) \|_{\B(L^2_s(\R^3), L^2_{-s} (\R^3))} \| \phi_\nn f \|_{L^2_s (\R^3)}  \\
& \leq c  \sum_{ |\nn| \leq n_0} \| r_0(\mu+i\ve  ) - r_0 (\mu) \|_{\B(L^2_s(\R^3), L^2_{-s} (\R^3))} \| \langle x\rangle^s  \phi_\nn \|_{L^\infty (\R^3) } \|  f \|_{L^2 (\R^3)} .
\end{align}

\n
For the high-energy part, we  recall that in the proof of Proposition \ref{LAPR0} we  showed that $ R_0^> (\mu+i\ve )$ is analytical with values in $\B(L^2 (\R^6), H^2 (\R^6))$ and therefore 
$\Tr R_0^> (\mu+i\ve )$ is analytical with values in $\B(L^2 (\R^6), L^2 (\R^3))$. By duality $\G^> (\mu+i\ve)$ converges to  $\G^> (\mu)\in \B(L^2 (\R^3), L^2 (\R^6))  $.
\end{proof}

\vs

\n
We now approach the study of the boundary limits of $(\Gamma (z) +\al )^{-1}$. 

\begin{prop} \label{LAPGamma}
There exists a set $\NN$ such that for 
$\mu \notin \NN$ the following limits exist in the uniform topology of  $\B(L^2(\R^3) )$
\[
\lim_{\ve \downarrow 0} (\Gamma (\mu \pm i \ve) +\al )^{-1} = (\Gamma^\pm (\mu) +\al )^{-1}
\]
and the limits are continuous in $\mu$. Moreover we have
$\NN \cap (\R^+\setminus \N) \subset \E$.
\end{prop}
\begin{proof}

We first discuss the boundary limit of $K(\mu+i\ve)$, defined in \eqref{Kappa}. Again we adopt a decomposition into a low-energy term and a high-energy term, that is we have
\beq \label{drugo}
K(\mu+i\ve) = K^<(\mu+i\ve) + K^>(\mu+i\ve)= \G^\ast(-\la ) \G^<(\mu+i\ve) +  \G^\ast(-\la ) \G^>(\mu+i\ve).
 \eeq
Since $ \G^> (\mu+i\ve )$ converge to  $ \G^> (\mu )$ as $\ve \to 0$ in $\B (L^2 (\R^3) , L^2(\R^6))$, see Proposition  \ref{LAPpot}, it follows that
$ K^>(\mu+i\ve)$ converge to 
\beq
 K^>(\mu)=  \G^\ast(-\la ) \G^>(\mu) \in  \B_\infty (L^2 (\R^3) ),
\eeq
and $ K^>(\mu)$ is a $\mu$ continuous valued compact operator. 
We claim that an analogous statement holds for $ K^<(\mu+i\ve)$ and  $ K^{+,<}(\mu):=  \G^\ast(-\la ) \G^{+,<}(\mu)$. 
Taking into account the orthogonality of the eigenvectors $\phi_\nn$,
the integral kernel of $K^<(\mu+i\ve)$ and $ K^{+,<}(\mu)$  are  given by:
\begin{align*}
K^<(\mu+i\ve; \y' , \y'')
& =  \sum_{ |\nn| \leq n_0} \phi_\nn (\y') \phi_\nn (\y'') \int_{\R^3} \frac{e^{i\sqrt{\mu - \ome |\nn| +i\ve}|\x-\y'|}  }{4\pi |\x-\y'| }  \frac{e^{-\sqrt{\la} |\x-\y''|}  }{4\pi |\x-\y''| } d\x \\
K^{<,+}(\mu; \y' , \y'')
& =  \sum_{ |\nn| \leq n_0} \phi_\nn (\y') \phi_\nn (\y'') \int_{\R^3} \frac{e^{i\sqrt{\mu - \ome |\nn| }|\x-\y'|}  }{4\pi |\x-\y'| }  \frac{e^{-\sqrt{\la} |\x-\y''|}  }{4\pi |\x-\y''| } d\x .\\
\end{align*}
Moreover, using elliptic coordinates, the integral in $\x$ can be computed explicitly.
For $b \in \R$, $b>0$, and $|\Im a|< b $ we define
\beq 
\mu = \frac{ |\x-\y'|+ |\x-\y''|}{ |\y'-\y''|} \in [1,+\infty)\,,   \qquad \nu = \frac{ |\x-\y'|- |\x-\y''|}{|\y'-\y''|} \in [-1,1]
\eeq
and $\varphi \in [0,2\pi)$  the rotation angle around the axis $\y' \y''$. Then, by  straightforward calculations, we have
\begin{align*}
\int_{\R^3} \frac{e^{ia|\x-\y'|}  }{ |\x-\y'| }  \frac{e^{-b |\x-\y''|}  }{ |\x-\y''| } d\x 
&=\int_1^{\infty} d\mu \int_{-1}^1 d\nu \int_0^{2\pi} d\varphi \, \frac{|\y'- \y''|^3}{8}\,  (\mu^2-\nu^2) \, \frac{ e^{i a \frac{|\y'- \y''|}{2}(\nu+\mu) } e^{ -b  \frac{|\y'- \y''|}{2}(\mu-\nu)    } }{ \frac{|\y'- \y''|}{2}(\nu+\mu) \frac{|\y'- \y''|}{2}(\mu-\nu) }\\
&= 4\pi \frac{1}{ |\y'- \y''|}\frac{   e^{ia    |\y'- \y''|} - e^{-b|\y'- \y''|} }{a^2 + b^2}.
\end{align*}
Then $K^<(\mu) $ and $K^{<,+}(\mu)$ read:
\begin{align*}
K^<(\mu+i\ve; \y' , \y'')
& =  \sum_{ |\nn| \leq n_0} \phi_\nn (\y') \phi_\nn (\y'')  \f{1}{4\pi |\y'-\y''| }  \f{e^{i\sqrt{\mu - \ome |\nn| +i\ve}|\y'-\y''|}    -  e^{-\sqrt{\la} |\y'-\y''|}  }{ \la + \mu - \ome |\nn| +i\ve } \\
K^{<,+}(\mu; \y' , \y'')
& =  \sum_{ |\nn| \leq n_0} \phi_\nn (\y') \phi_\nn (\y'')  \f{1}{4\pi |\y'-\y''| } \f{e^{i\sqrt{\mu - \ome |\nn| }|\y'-\y''|}    -  e^{-\sqrt{\la} |\y'-\y''|}  }{ \la + \mu - \ome |\nn|  }.\\
\end{align*}
We notice that
\beq \label{ricottaspinaci}
\lf|  \f{1}{4\pi |\y'-\y''| }  \f{e^{i\sqrt{\mu - \ome |\nn| +i\ve}|\y'-\y''|}    -  e^{-\sqrt{\la} |\y'-\y''|}  }{ \la + \mu - \ome |\nn| +i\ve } \ri| \leq c.
\eeq
Using \eqref{ricottaspinaci}, it is straightforward to prove that $K^<(\mu+i\ve )$ and $K^{<,+} (\mu )$ are Hilbert-Schmidt operators and that $\| K^<(\mu+i\ve) \|_{HS}$ is  uniformly bounded in $\ve$.  Moreover we have
\begin{multline}
\| K^<(\mu+i\ve ) - K^<(\mu ) \|_{HS}^2  \leq c  \sum_{ |\nn| \leq n_0}\int_{\R^6} d\y' \, d\y'' \, \left|  \phi_\nn (\y') \phi_\nn (\y'') \ri|^2 \f{1}{16\pi^2 |\y'-\y''|^2 }\\
\lf| \f{e^{i\sqrt{\mu - \ome |\nn| +i\ve}|\y'-\y''|}    -  e^{-\sqrt{\la} |\y'-\y''|}  }{ \la + \mu - \ome |\nn| +i\ve }-  \f{e^{i\sqrt{\mu - \ome |\nn| }|\y'-\y''|}    -  e^{-\sqrt{\la} |\y'-\y''|}  }{ \la + \mu - \ome |\nn|  }    \ri|^2.
\end{multline}
Hence $\| K^<(\mu+i\ve ) - K^{<,+}(\mu ) \|_{HS}\to 0$ as $\ve\to 0$ by dominated convergence.
Again by dominated convergence, it is straightforward to see that $ K^{<,+}(\mu) $ is continuous in $\mu$. 
Then we obtain    by construction the  boundary values of $\Ga (z) $ and $M(z)$ 
\beq \label{scrivomale}
\Gamma^\pm(\mu) = \Gamma (-\la) - (\la+\mu)  K^\pm (\mu) \qquad 
M^\pm (\mu) =  (\la+\mu)  (\Gamma(-\la) + \al )^{-1/2} K^\pm (\mu)  (\Gamma(-\la) + \al )^{-1/2} .
\eeq
Moreover  $M^\pm (\mu) \in  \B_\infty (L^2 (\R^3) ) $ and they are continuous in $\mu$ in the uniform topology.

\n
Let us now define 
\[
\NN=\{ \mu\in[0,\infty), \, | \exists f\in L^2(\R^3),\, f\neq 0,\,  (I-M^+ (\mu))f=0 \}.
\]
For $\mu \notin \NN$, by Fredholm theory, $(I-M(\mu+i\ve))^{-1}$ converge to  $(I-M^+(\mu))^{-1}$  in $\B(L^2(\R^3))$  and $(I-M^+(\mu))^{-1}$ is continuous in $\mu$. 
Therefore we can take the boundary values of \eqref{waraku} and we obtain
\beq \label{trenitalia}
(\Gamma^\pm (\mu) + \al )^{-1} =  (\Gamma(-\la) + \al )^{-1/2}(I -M^\pm (\mu))^{-1} (\Gamma(-\la) + \al )^{-1/2} \qquad \mu \notin\NN.
\eeq

\n
Now we prove that $\NN \cap (\R^+\setminus \N) \subset \E$ (recall that $\E$ denotes the set of positive eigenvalues of $H$ of finite multiplicity).
Let us consider $\mu\in \NN$ such that $\mu \notin \N$. Then we have 
\[
\| f\|_{ L^2(\R^3) }^2 = \langle f , M^+ (\mu)f \rangle_{ L^2(\R^3) }
\]
which implies 
\beq \label{cancellation1}
 \text{Im} \langle f , M^+ (\mu)f \rangle_{ L^2(\R^3) }=0.
\eeq
Moreover, since $K^> (\mu) $ is self-adjoint, the previous identity reduces to
\[
0
=  \lf\langle  (\Gamma(-\la) + \al ) ^{-1/2} f, (  K^{<,+}(\mu)-   K^{<,-}(\mu))  (\Gamma(-\la) + \al ) ^{-1/2} f \ri\rangle_{ L^2(\R^3) },
\]
with
\begin{align*}
 (  K^{<,+}(\mu)-   K^{<,-}(\mu)) ( \y' , \y'')
 & =  \sum_{ |\nn| \leq n_0} \phi_\nn (\y') \phi_\nn (\y'') \int_{\R^3} \lf( \frac{e^{i\sqrt{\mu - \ome |\nn| }|\x-\y'|}  }{4\pi |\x-\y'| }-  \frac{e^{-i\sqrt{\mu - \ome |\nn| }|\x-\y'|}  }{4\pi |\x-\y'| }   \ri)   \frac{e^{-\sqrt{\la} |\x-\y''|}  }{4\pi |\x-\y''| } d\x \\
 & =  \sum_{ |\nn| \leq n_0} \phi_\nn (\y') \phi_\nn (\y'') \int_{\R^3} 2i\, \text{Im}\, r^+_0 (\mu-\ome|\nn|, \x-\y')   \frac{e^{-\sqrt{\la} |\x-\y''|}  }{4\pi |\x-\y''| } d\x .
\end{align*}
Using Fourier transform, \eqref{agmontrace} and \eqref{cancellation1} imply
\beq \label{papa}
\lf. \widehat{ \phi_\nn  (\Gamma(-\la) + \al ) ^{-1/2} f }\ri|_{|\k|= \sqrt{\mu-\ome |\nn|} }=0 \qquad \qquad  |\nn|\leq n_0.
\eeq
Let us define $g=  (\Gamma(-\la) + \al ) ^{-1/2} f\in L^2(\R^3) $ and $\psi = \G^+ (\mu) g$. We claim   that $\psi \in L^2(\R^6)$ and it is an eigenvector of $H$ with eigenvalue $\mu$. 

\n
We write $\psi = \G^{+,<} (\mu) g + \G^> (\mu) g$. Since $g\in L^2(\R^3)$ then $\| \G^> (\mu) g \|_{ L^2(\R^6)}\leq c \|g\|_{L^2(\R^6)}$ due to the previous arguments. 
While in general we do not expect $\G^{+,<} (\mu) g$ to be in $\in L^2(\R^6)$, condition \eqref{papa} provides a crucial cancellation.
We recall that (see \cite{R} Remark 4.1 and \cite{BD})  if  we have  $\text{Im}\, r_0^+ (\mu) f =0$, with $f \in L^2_s (\R^3)$, $s>1/2$, and  $\mu>0$, 
 then $ r_0^+ (\mu) f  \in L^2 (\R^3)$ and $ \|  r_0^+ (\mu) f  \|_{ L^2 (\R^3)} \leq c \|f \|_{ L^2_s (\R^3)}$. 
Condition \eqref{papa} ensures that all the terms in  $ \G^{+,<} (\mu) g$ are in fact in $ L^2(\R^6)$ and we have
\[
 \| \G^{+,<} (\mu) g \|_{  L^2(\R^6)}\leq c     \sum_{ |\nn| \leq n_0} \| \phi_\nn g\|_{L^2_s (\R^3)} \leq c  \| g\|_{L^2 (\R^3)}.  
 \]
Collecting together the estimates for the low-energy part and the high-energy part, we obtain that \eqref{papa} implies 
\beq\label{mama}
 \| \G^{+} (\mu) g \|_{  L^2(\R^6)}\leq c     \| g\|_{L^2 (\R^3)}.
 \eeq
It remains to prove that $\psi \in \D (H) $ and $H\psi =\mu \psi$. First we write
\beq 
\psi = \G^{+} (\mu) g = \G^{+} (\mu) g - \G(-\la) g +\G(-\la) g \equiv \phi^\la + \G(-\la) g.
\eeq
The resolvent identity
\beq \label{holter}
\G^{+} (\mu) g - \G(-\la) g = (\la+\mu) R_0(-\la)  \G^{+} (\mu) g
\eeq
shows that $\phi^\la \in \D (H_0)$ since $ \G^{+} (\mu) g\in L^2(\R^6)$. Moreover $\psi $ satisfies the boundary condition $(\Ga(-\la) +\al ) g = \Tr \phi^\la$
since $\Tr \phi^\la= \Ga(-\la) g - \Ga^{+} (\mu) g$ and $( \Ga^{+} (\mu) +\al )g=0$ by construction, hence $\psi \in \D(H)$.
Finally we have
\begin{align*}
H \psi 
& = H_0 \phi^\la - \la \G(-\la) g \\
& =  (\la+\mu) H_0 R_0(-\la)  \G^{+} (\mu) g - \la \G(-\la) g \\
& =  (\la+\mu) (I- \la R_0(-\la)  )  \G^{+} (\mu) g - \la \G(-\la) g \\
&= \mu   \G^{+} (\mu) g +\la \lf( \G^{+} (\mu) g - \G(-\la) g - (\la+\mu) R_0(-\la)  \G^{+} (\mu) g \ri)\\
&= \mu   \G^{+} (\mu) g ,
\end{align*}
where in the last step we used \eqref{holter}. Therefore we have $\NN \cap (\R^+\setminus \N) \subset \E$.
\end{proof}

\vs

\n
We are now in position to prove Theorem \ref{LAPmain}. 
\begin{proof}

The boundary values of $R(z)$ on the real line exist in $\B (L^2_{s} (\R^3) , L^2_{-s} (\R^3)) \otimes \B (L^2(\R^3))$ and they are continuous in $\mu$ for $\mu \notin \NN$,
 by Proposition \ref{LAP0}, Proposition \ref{LAPpot} and Proposition \ref{LAPGamma}. 
 
 \n
It remains to prove is that $\E$ is discrete. 
Suppose by contradiction that $\E$ is not discrete, then there exists a sequence of positive eigenvalues $\{ \mu_n \}$ such that $\mu_n\to \mu$ and, using the same notation as above, there exists also $\{f_n \}$, $\{g_n\}$, $\{ \psi_n \}$ such that
\[
\psi_n =  \G^+ (\mu) g_n = \G^+ (\mu)  (\Gamma(-\la) + \al ) ^{-1/2} f_n.
\]
We do not know a priori if $\mu$ is a positive eigenvalue. We can assume that $\| f_n\|_{L^2(\R^3) } =1$, then there exists a subsequence, still denoted by $\{f_n\}$, such that $f_n \rightharpoonup f$.
Each $f_n$ satisfies 
\[
f_n = -\lf ( K^+ (\mu_n) - K^+(\mu) \ri) f_n -  K^+(\mu) f_n.
\]
The first term on the r.h.s. converges strongly to zero  by the continuity of $K^+(\mu)$, since $\| f_n\|_{L^2(\R^3) } =1$.
The second term on the r.h.s. converges strongly since  $K^+(\mu)$ is compact. Therefore we obtain that $f_n \to f$ in $L^2(\R^3)$. Then  $\| f\|_{L^2(\R^3) } =1$ and $f\neq 0$. Moreover $g_n \to g = (\Gamma(-\la) + \al ) ^{-1/2} f$ and $\|g\|_{L^2(\R^3) } \leq c $.
Furthermore $\| \psi_n \|_{L^2(\R^6) } \leq c $ by \eqref{mama}. 
Since the collection of the eigenvectors $\psi_n$ 
is an orthogonal set, we have $\psi_n \rightharpoonup 0$ which is absurd. Then $\E$ admits no accumulation point and it is discrete. 
The same argument shows that these positive eigenvalues must have finite multiplicity and therefore the proof of Theorem \ref{LAPmain} is complete. 
\end{proof}

\vs
\n
Further results on the spectral and the scattering properties of $H$ are direct   consequences of Theorem \ref{LAPmain}. 

\begin{cor}
$\;\sigma_{ac}(H) = [0,\infty)$, $\;\sigma_{sc}(H)=\emptyset\;$   and the wave operators exist and are complete.
\end{cor}
\begin{proof}
All these statement follow from standard results, see \cite{RS}. Notice that part of these results were proved also in \cite{CDF} by different techniques.
\end{proof}

\vs

\begin{rem}
In the proof of Theorem \ref{LAPmain} we have not characterized the singular values $\mu\in \NN\cap \N$. If we had a non trivial solution of $ (I+M^+ (\mu))f=0$ for such $\mu$, and we would repeat the previous arguments
we would not be able to prove that $\psi \in L^2 (\R^6)$. In fact,  the terms in $\langle f,  M^+ (\mu) f\rangle$ with $ |\nn|=n_0 = \mu/\ome$ would be real and therefore \eqref{papa} would hold for $ |\nn|<n_0$. 
In other words,  we would be not able to infer that
\[
\left. \widehat{ \phi_\nn  (\Gamma(-\la) + \al ) ^{-1/2} f }\ri|_{|\k|= \sqrt{\mu-\ome|\nn|} }=0 \qquad \qquad  |\nn|= n_0.
\]
This would imply that the corresponding terms in $\G^+ (\mu) g $ have a slower decay in the $\x$ variable and do not belong in $L^2 (\R^6)$. In particular they would decay at infinity like $|\x|^{-1}$
 like a zero energy resonance in the one-body case. 
\end{rem}
\begin{rem}
For the one-body case, the proof that positive eigenvalues do not exists relies on some continuation properties of $-\Delta$, see \cite{IJ}, or  some properties of the virial in conjunction with Mourre theory, see \cite{FH}.
Such technical results do not apply to $H_0$ due to the presence of the harmonic hoscillator. The extension of such results to our case is outside the scopes of this paper.
\end{rem}

\vs

\section{Stationary scattering theory}
In this section we develop the stationary scattering approach for $H$. We follow and adapt \cite{K} (for general results in the case of singular interactions, studied with different techniques,  see also \cite{MP1,MP2,MP3}). 
%
For sake of notation, in the following we denote
\[
E=E (\k, \nn) = |\k|^2 + \ome |\nn| \qquad E'=E  (\k', \nn') =  {|\k'|^2} + \ome |\nn'| .
\]
Then  we can cast \eqref{borbotto} 
 in a more compact form
\begin{align}
\Phi_\pm (\k, \nn) & = \Phi_0 (\k, \nn) + \G^\pm (  E(\k,\nn)  ) ( \Ga^\pm (  E(\k,\nn) ) + \al )^{-1} \Tr \Phi_0 (\k, \nn) \\
&= \lf( I  + \G^\pm ( E(\k,\nn) ) ( \Ga^\pm (  E(\k,\nn) ) + \al )^{-1} \Tr \ri)  \Phi_0 (\k, \nn) \label{bari}.
\end{align}
Since $\Tr \Phi_0 (\k, \nn)\in L^2 (\R^3)$, we have that  $\Phi_\pm (\k, \nn)$ is well defined by Proposition \ref{LAPGamma} and it belongs to $ L^2_{-s} (\R^3) \otimes  L^2(\R^3)$ for $s>3/2$. 
For later convenience we define in $L^2 (\R^6)$ 
\begin{align*}
Q(z) &=  I  + \G (z   ) ( \Ga (  z) + \al )^{-1} \Tr    \qquad \text{Im } z\neq 0\\
Q^\pm (\mu) &=  I  + \G^\pm (  \mu ) ( \Ga^\pm (  \mu) + \al )^{-1} \Tr \qquad \mu>0 \;\; \mu\notin \NN.
\end{align*}


\n
We also define  the formal adjoint operators w.r.t. the inner product of $L^2(\R^6)$
\begin{align*}
Q^\ast (z) &=  I  + \de_\pi ( \Ga^\ast (  z) + \al )^{-1}  \G (z   )   \qquad \text{Im } z\neq 0\\
Q^{\pm, \ast} (\mu) &=  I  +\de_\pi ( \Ga^{\pm,\ast} (  \mu) + \al )^{-1}  \G^{\pm ,\ast} (  \mu )  \qquad  \mu>0 \;\; \mu\notin \NN.
\end{align*}
Some remarks are in order. From the definition, we have $Q(z): H^s (\R^6)\mapsto L^2(\R^6)$ and  $Q^\ast (z):  L^2(\R^6)\mapsto  H^{-s} (\R^6)$, $s>3/2$ as bounded operators due to trace theorems,
but the results of the previous section, see in particular Proposition \ref{LAPR0} and Lemma \ref{LAPpot}, show that in fact $Q(z)$ is actually defined on functions with trace over $\pi$ in $ L^2(\R^3)$  with values in $ L^2(\R^6)$, like $\Phi_0$, and
for $Q^\ast (z)$ a dual statement holds. In the following we will often use this property.
For the boundary values, $Q^{\pm} (\mu) $ is defined on function with trace over $\pi$ in $ L^2(\R^3)$  with values in $ L^2_{-s}(\R^3)\otimes L^2(\R^3)$
and dual statement holds for $Q^{\pm, \ast} (\mu) $. Moreover,  they are continuous in $\mu$ when acting on such functions.
With this notation we can write
\beq \label{faenza}
\Phi_\pm (\k, \nn) = Q^\pm ( E(\k,\nn) ) \Phi_0 (\k, \nn) 
\eeq
and
\begin{align}
& R(z) = Q(z)\, R_0(z) \qquad \text{Im } z\neq 0 \label{prob1} \\
& R^\pm (\mu) = Q^\pm (\mu) \, R^\pm_0 (\mu) \qquad  \mu>0 \;\; \mu\notin \NN. \label{prob2}
\end{align}

\n
In the next Proposition we prove a crucial identity which relates $\Phi_\pm (\k, \nn)$ and $R^\pm ( \mu )$.
\begin{prop}
Let  $\mu>0\,,$ with  $\; \mu\notin \NN$. Then following identity holds in $\S'(\R^6) \otimes \S'(\R^6)$
\beq \label{volta}
\frac{1}{2\pi i } \lf( R^+ (\mu ) -  R^- (\mu) \ri) =
\sum_{ |\nn|<n_0}  \frac{|\k_n|}{2} \int_{\mathbb S^2}  \Phi_\pm (\k_n , \nn) \otimes \overline{ \Phi_\pm (\k_n , \nn) } d\boldsymbol{\ome} 
\eeq
 where 
\[
\k_n= \sqrt{\mu-\ome |\nn|} \, \boldsymbol{\ome}.
\]
\end{prop}
\begin{proof}
We start from
\beq \label{level1}
\frac{1}{2\pi i } \lf( r_0^+ (\mu) -  r_0^- (\mu) \ri)=  \f{\sqrt{\mu}}{2}\,  \int_{\mathbb S^2}   e(\k) \otimes \overline{  e(\k) } d\boldsymbol{\ome}  \qquad \k= \sqrt{\mu} \boldsymbol{\ome} \qquad e(\k,\x)  =\frac{1}{(2\pi)^{3/2}} e^{i\k\cdot \x} 
\eeq
where 
identity \eqref{level1} must be understood in $\S'(\R^3) \otimes \S'(\R^3)$ and it is just a rephrasing of \eqref{agmontrace} keeping into account the rescaling of the surface measure.
Let us prove that 
\beq \label{level2}
\frac{1}{2\pi i } \lf( R_0^+ (\mu ) -  R_0^- (\mu) \ri) = \sum_{ |\nn|<n_0}  \f{|\k_n|}{2} \int_{\mathbb S^2} 
\Phi_0 (\k_n , \nn) \otimes \overline{ \Phi_0 (\k_n , \nn) } d\boldsymbol{\ome}    
\eeq
in  $\S'(\R^6) \otimes \S'(\R^6)$. 
Let us notice that, see the final part of the proof of Proposition \ref{LAPR0}, 
\begin{align*}
\frac{1}{2\pi i } \lf( R_0^+(\mu)  -  R_0^- (\mu) \ri) & = \frac{1}{2\pi i } \lf( R_0^{<,+}(\mu)  -  R_0^{<,-} (\mu) \ri) \\
& = \sum_ {|\nn|<n_0} \frac{1}{2\pi i } ( r_0^+ (\mu - \ome |\nn|) -  r_0^- (\mu - \ome |\nn|)  ) \otimes P_\nn
\end{align*}
since the high-energy parts cancels out. Using  \eqref{level1}, we obtain \eqref{level2}. We observe  that in the sum we can substitute $\{  \ome |\nn| \leq \mu \}$ with $\{  \ome |\nn| < \mu \}$ since the terms with $\ome |\nn|= \mu$ are null. 

\n
We now prove \eqref{volta}.  
Using the resolvent identity, we have
\begin{align*}
\frac{1}{2\pi i } \lf( R^+ (\mu ) -  R^- (\mu) \ri) 
&= \frac{1}{2\pi i } \lim_{ \ve \downarrow 0} \lf( R (\mu+i \ve ) -  R (\mu-i\ve) \ri) \\
&=\frac{1}{\pi  }  \lim_{ \ve \downarrow 0} \ve \,  R (\mu+ i \ve )  \, R (\mu- i\ve)  \\
&=\frac{1}{\pi  }  \lim_{ \ve \downarrow 0} \ve \, Q (\mu+ i \ve )  R_0 (\mu+ i \ve )  \, R_0 (\mu- i\ve) Q^\ast (\mu-  i \ve )  \\
&=\frac{1}{2\pi i } \lim_{ \ve \downarrow 0}  Q (\mu+ i \ve ) \lf( R_0 (\mu+i \ve ) -  R_0 (\mu-i\ve)  \ri)  Q^\ast (\mu+  i \ve ).
\end{align*}
The same argument can be repeated exchanging signs. We have already discussed the convergence of $(\Ga(\mu \pm i\ve) +\al)^{-1}$ to $(\Ga^\pm(\mu ) +\al)^{-1}$ in Proposition \ref{LAPGamma} and 
the convergence of $\G (\mu \pm i\ve) $ of $\G^\pm (\mu)$ in Proposition \ref{LAPpot}. The 
 delicate point is the presence of the traces in $Q$ and $Q^\ast$,
because we know that  $ R_0 (\mu\pm i \ve )$ converge only in $\B (L^2_{s} (\R^3) , L^2_{-s} (\R^3)) \otimes \B (L^2(\R^3))$. 
In fact,  it is sufficient to prove that also the traces of  $ R_0 (\mu\pm i \ve )$ converge. We write
\begin{align*}
 R_0 (\mu+i \ve ) -  R_0 (\mu-i\ve)  & =  R_0^< (\mu+i \ve ) -  R_0^< (\mu-i\ve)  +2 \ve  R_0^> (\mu+i \ve )  R_0^> (\mu-i\ve) \\
&=  \sum_{ |\nn|<n_0} \lf( r_0(\mu+i\ve -\ome |\nn|)-  r_0(\mu-i\ve -\ome |\nn|)\ri)  \otimes P_\nn  \\ 
&+ 2i \ve  R_0^> (\mu+i \ve )  R_0^> (\mu-i\ve).
\end{align*}
Let us focus on the last term on the r.h.s..
We known that $R_0^> (\mu\pm i\ve)$ converges in $\B(L^2 (\R^6), H^2 (\R^6))$ and by duality it converges also in  $\B(H^{-2} (\R^6), L^2 (\R^6) )$.
Therefore the product converges (and it is bounded in the uniform topology) in $\B(H^{-2} (\R^6), H^2 (\R^6) )$ and its contribution vanishes due to the $\ve$ in front.

\n
The low-energy part $R_0^< (\mu+i \ve )  -  R_0^< (\mu-i\ve)$ is a finite sum of operators, with integral kernel of the form $A_\ve (\x,\y,\x',\y')$. 
Therefore it is sufficient to prove that $A_\ve (\x,\x,\x',\x')$ converges in $\B(L^2 (\R^3))$.
In fact, they are of the form $\phi_\nn \, r_0 (\nu \pm i \ve )\, \phi_\nn$, $\nu>0$ and we have
\[
\| \phi_\nn \lf(  r_0 (\nu \pm i \ve )- r_0^\pm (\nu) \ri) \phi_\nn \|_{ \B(L^2 (\R^3) ) } \leq c  \| \langle x \rangle^s \phi_\nn\|^2_{L^\infty(\R^3)} \| r_0 (\nu \pm i \ve )- r_0^\pm (\nu)\|_{\B(L^2_s (\R^3), L^2_{-s} (\R^3))}
\]
and the r.h.s. goes to zero by Proposition \ref{LAP0}. 
Therefore we have
\begin{align*}
\frac{1}{2\pi i } \lf( R^+ (\mu ) -  R^- (\mu) \ri) 
&=\frac{1}{2\pi i }  Q^\pm (\mu ) \lf( R_0^+ (\mu ) -  R_0^- (\mu)   \ri)Q^{\pm,\ast} (\mu ) \\
& =  \sum_{ |\nn|<n_0}  \f{|\k_n|}{2} \int_{\mathbb S^2}   Q^\pm (\mu ) \Phi_0 (\k_n , \nn) \otimes \overline{ \Phi_0 (\k_n , \nn) } Q^{\pm,\ast} (\mu )   d\boldsymbol{\ome}    \\
&=  \sum_{ |\nn|<n_0}  \f{|\k_n|}{2} \int_{\mathbb S^2}  \Phi_\pm (\k_n , \nn) \otimes \overline{ \Phi_\pm (\k_n , \nn) } d\boldsymbol{\ome} ,
\end{align*}
and this concludes the proof. 
\end{proof}

\vs

\n
We are now in position to prove Theorem \ref{eigexp}. 

\begin{proof}
We know that $\R^+\setminus \NN= \bigcup_j I_j$, where $I_j$ are open intervals. Let us fix an interval $I\subset I_j$ and $u\in \S (\R^6)$.  Then by Stone's formula and Proposition \ref{volta} we have
\begin{align}
\langle E_I (H) u, u \rangle  
&= \frac{1}{2\pi i } \int_I \langle \lf( R^+ (\mu ) -  R^- (\mu) \ri) u, u\rangle d\mu \nonumber \\
&= \int_I \sum_{ |\nn|<n_0} \f{|\k_n (\mu)|}{2}  \int_{\mathbb S^2} \lf| \F_\pm u ( \k_n (\mu), \nn)\ri|^2 d\boldsymbol{\ome} \, d\mu \nonumber\\
&= \int_{\R^3}  \sum_{\nn\in\N_0 }\chi_{\{{k^2+\ome |\nn|} \in I\} } \lf| \F_\pm u ( \k ,\nn)\ri|^2 d\k. \label{extension}
\end{align}
In the last step we changed variable in the integration over $\mu$, introducing $k$ defined by $k^2= \mu - \ome |\nn|$ and $\k= k \, \boldsymbol{\ome}$. Notice that the series over $\nn$ is actually a 
finite sum due to the indicator function.
Since $| \langle E_H (I) u, u \rangle  | \leq \| u \|^2_{L^2(\R^6)} $, we can extend \eqref{extension} in a twofold way: we let $I$ invade $I_j$ and then sum over $j$ and we extend the domain from $\S (\R^6)$ to 
$L^2(\R^6)$. Since the set $\{ k^2+\ome |\nn|\in  \NN\}$ has zero measure in $\R^3$ for $\forall \nn \in \N_0^3$, we obtain  $\F_\pm : L^2(\R^6) \mapsto L^2(\R^3)\otimes \ell^2 (\N_0^3)$ as a bounded operator and 
\beq
\| \F_\pm u \|_{L^2(\R^3)\otimes \N_0^3} = \| E_H ( \R^+ \setminus \NN) u \|_{L^2(\R^6)}.
\eeq

\noindent
Let us now prove \eqref{call1} and \eqref{call2}. Let us fix $I$ compact interval, $I\subset I_j$,  and
let us consider $u$ such that $(\F_0 u)(\k, \nn) =\tilde u _\nn (\k)$ is different from zero for a finite number of $\nn$, $ \tilde u _\nn (\k) \in \C^{\infty}_0 (\R^3) $ and $\text{supp} \,\tilde u_\nn \subset\{ |\k | \in  I-\ome |\nn|\}$.
In particular, we have $E_{H_0} (I) u = u$ and $u \in \S (\R^6)$. Let us also consider  $ v \in \S(\R^6)$. Using the standard stationary representation of the wave operator $W_+$, see \cite{K}, we have
\begin{align*}
\langle W_+ u , v \rangle 
& = \lim_{\ve \downarrow 0} \frac{\ve}{\pi} \int_I \langle R(\mu -i \ve ) \, R_0 (\mu + i \ve ) u, v \rangle d\mu\\
& = \lim_{\ve \downarrow 0} \frac{\ve}{\pi} \int_I \langle Q(\mu -i \ve )\, R_0(\mu -i \ve ) \, R_0 (\mu + i \ve ) u, v \rangle  d\mu \\
& = \lim_{\ve \downarrow 0} \frac{\ve}{\pi} \int_I \langle  R_0(\mu -i \ve ) \, R_0 (\mu + i \ve ) u,  Q^\ast (\mu -i \ve )\, v \rangle  d\mu \\
& = \lim_{\ve \downarrow 0} \frac{1}{2 \pi i } \int_I \langle \lf(  R_0 (\mu + i \ve )- R_0(\mu -i \ve )\ri)  u,  Q^\ast (\mu -i \ve )\, v \rangle  d\mu .
\end{align*}
In order to take the limit in $\ve$ under the integral it is sufficient to repeat the arguments used in Proposition \ref{volta}, to notice that the limit is a continuous function of $\mu$ and to recall that we are integrating over a compact set.

\n
Using \eqref{level2}, \eqref{faenza} and the definition of $\F_0$ and $\F_+$, we have
\begin{align*}
\langle W_+ u , v \rangle 
& = \frac{1}{2 \pi i } \int_I \langle \lf(  R_0^+ (\mu  )- R_0^+(\mu )\ri)  u,  Q^{+,\ast} (\mu  )\, v \rangle  d\mu  \\
&=\int_I d\mu \, \sum_{ |\nn|<n_0} \f{|\k_n|}{2} \int_{\mathbb S^2}     \overline{\langle \Phi_0 (\k_n , \nn), u\rangle} \langle  \Phi_0(\k_n , \nn) , Q^{+,\ast} (\mu  )\, v \rangle  d\boldsymbol{\ome}  \\
&=\int_I d\mu \, \sum_{ |\nn|<n_0}\f{|\k_n|}{2} \int_{\mathbb S^2}     \overline{\langle \Phi_0 (\k_n , \nn), u\rangle} \langle  \Phi_+ (\k_n , \nn) ,\, v \rangle  d\boldsymbol{\ome}  \\
&=\int_I d\mu \, \sum_{ |\nn|<n_0} \f{|\k_n|}{2} \int_{\mathbb S^2} \overline{ (\F_0 u)(\k_n , \nn) } (\F_+ \, v)(\k_n , \nn)  d\boldsymbol{\ome}  \\
&= \int_{\R^3}  \sum_{\nn\in\N_0 }\chi_{\{{k^2}+\ome |\nn|\in I\} } \overline{ (\F_0 u)(\k_n , \nn) } (\F_+ v) ( \k ,\nn) d\k \\
&= \langle \chi_I \F_0 u, \F_\pm v \rangle_{ L^2(\R^3)\otimes \ell^2 (\N_0^3)}.
\end{align*}
In the last step we changed variable in the integration over $\mu$, introducing $k$ defined by $k^2= \mu - \ome |\nn|$ and $\k= k \, \boldsymbol{\ome}$. We can again let $I$ invade $I_j$ and sum over $j$. Since 
the set of functions $u$ such that $\F_0 u$ has support away from $\NN$ is dense in $L^2(\R^6)$, we obtain 
\beq
\langle W_+ u , v \rangle = \langle  \F_0 u, \F_+ v \rangle_{ L^2(\R^3)\otimes \ell^2 (\N_0^3)}.
\eeq
The same argument can be repeated with minor modifications for $W_-$ and \eqref{call1} is proved. 
Properties \eqref{call2} immediately follow from the fact that  the wave operators exist and are complete. This concludes the proof of Theorem  \ref{eigexp}.
\end{proof}

\vs

\section{Born Approximation}
\noindent
In this section we show how one can formally recover Fermi's formula  \eqref{fermi} for the scattering cross-section in the Born approximation.  We plan to perform a more detailed perturbative analysis of the scattering cross-section beyond the Born approximation 
in a forthcoming  paper.

\n
Due to the result of the previous section we can cast the scattering matrix as
\[
S = W_+^\ast W_-=  \F_0^\ast \F_+ \F_-^\ast \F_0 \,.
\]
Passing  to the momentum representation, we write
\[
\hat S = \F_0 S \F_0^\ast =\F_+ \F_-^\ast ,
\]
and we focus on the $\T$-matrix
\[
\hat S = I -2 \pi i \T \qquad \quad   \T =\frac{1}{2\pi i} \F_+ \lf( \F_+^\ast - \F_-^\ast \ri).
\]
For $f,g \in L^2(\R^3)\otimes \ell^2 (\N_0^3)$ we have
\beq
\lf( \F_\pm^\ast f \ri) (\x , \y)= \int_{\R^3} d\k \sum_{\nn \in \N_0} \Phi_\pm ( \k, \nn ,\x,\y )  f(\k,\nn)  
\eeq
and
\beq
\langle g, \T f \rangle = \frac{1}{2\pi i}\int_{\R^6} d\k \, d\k' \sum_{\nn, \nn'} \overline{ g(\k,\nn)} \, f(\k' , \nn') \int_{\R^6}d\x\, d\y \, \overline{ \Phi_+ (\k, \nn ,\x, \y) } \lf(  \Phi_+ (\k', \nn' ,\x, \y) -  \Phi_- (\k', \nn' ,\x, \y) \ri).
\eeq
Therefore the integral kernel of the $\T$-matrix reads
\beq
\T(\k, \nn , \k', \nn')= \frac{1}{2\pi i} \int_{\R^6}d\x\, d\y \, \overline{ \Phi_+ (\k, \nn ,\x, \y) } \lf(  \Phi_+ (\k', \nn' ,\x, \y) -  \Phi_- (\k', \nn' ,\x, \y) \ri).
\eeq
For $\al \to +\infty$, $\langle g, \T f \rangle   \to 0$ and $\hat S\to I$, that is the scattering becomes trivial. 
The Born approximation consists in finding the first non trivial correction of $\T$ for $\alpha$ large.
Notice that the perturbative regime near $H_0$ is achieved for $\al\to +\infty$ not simply $\al \to \infty$ as for $\al\to -\infty$
the Hamiltonian admits infinitely many bound states and the lower bound of $\sigma (H)$ diverges to $-\infty$, see \cite{CDF}.
The crucial point, as \eqref{bari} suggests, is an expansion of $(\Ga^\pm (\mu) + \al)^{-1}$ in power of  $\alpha^{-1}$.
For this purpose, it is sufficient to consider the identity:
\beq \label{pilastro}
(\Ga^\pm (\mu) +\al )^{-1}F  = \frac{1}{\al }F  - \frac{1}{\al } (\Ga^\pm (\mu) +\al )^{-1} \Ga^\pm (\mu)  F \qquad F\in \D (\Ga (-\la) )
\eeq
which  gives
\beq 
\Phi_\pm (\k, \nn )\simeq \Phi_0 (\k, \nn )+ \frac{1}{\al} \G^\pm (E(\k,\nn) )\Tr \Phi_0 (\k, \nn ).
\eeq
In particular the leading term is
\begin{align*}
\T(\k, \nn , \k', \nn') &\simeq \frac{1}{2\pi i \al} \int_{\R^6}d\x\, d\y \, \overline{ \Phi_0 (\k, \nn ,\x, \y) } \lf( \G^+ (E' ) - \G^- (E' ) \ri)\Tr \Phi_0 (\k', \nn' )  \\
&=\frac{1}{2\pi i \al}  \int_{\R^6}d\x\, d\y \, \overline{ \Phi_0 (\k, \nn ,\x, \y) }  \lf( R_0^+ (E' ) - R_0^- (E' ) \ri)\de_\pi  \Phi_0 (\k', \nn' ).
\end{align*}
Writing explicitly all the integrals, we have
\begin{align*}
&\T(\k, \nn , \k', \nn')
 = \frac{1}{i (2\pi)^4 \al }   \int_{\R^{12}}d\x\, d\y \,d\x' \, d\y' \, e^{-i\k\cdot \x} \phi_\nn (\y)  \\
&  \;\;\;\;  \sum_{\mm< E'} \lf( r_0^+ ( E' -\ome|\mm| , \x-\x' ) - r_0^- ( E' -\ome|\mm| , \x-\x' ) \ri) \phi_\mm (\y) \phi_\mm (\y')  \de(\x' -\y') \, e^{i \k'\cdot \x'} \phi_{\nn'}(\y') \\
& = \frac{1}{i(2\pi)^4 \al }   \int_{\R^{9}}d\x\, d\y \,d\x' \,  \, e^{-i\k\cdot \x} \phi_\nn (\y)  \\
& \;\;\;\;  \sum_{\mm< E'} \lf( r_0^+ ( E' -\ome|\mm| , \x-\x' ) - r_0^- ( E' -\ome|\mm| , \x-\x' ) \ri) \phi_\mm (\y) \phi_\mm (\x')   \, e^{i \k'\cdot \x'} \phi_{\nn'}(\x') \\
& = \frac{1}{i (2\pi)^4\al }   \int_{\R^{3}}d\x \, e^{-i\k\cdot \x}  \int_{\R^3} \,d\x'   \, \lf( r_0^+ ( E' -\ome|\nn| , \x-\x' ) - r_0^- ( E' -\ome|\nn| , \x-\x' ) \ri)   e^{i \k'\cdot \x'}  \phi_\nn (\x')  \phi_{\nn'}(\x') \\
&= \frac{1}{\al (2\pi)^{3/2} } \frac{1}{2 \sqrt{E' -\ome|\nn| } } \de( k- \sqrt{E' -\ome|\nn| }   ) \widehat{\phi_\nn \phi_{\nn'} }(\k-\k') \\
& = \frac{1}{\al  (2\pi)^{3/2}} \de( k^2+ \ome|\nn| -k^{\prime 2} -\ome |\nn'|) \widehat{\phi_\nn \phi_{\nn'} }(\k-\k')
 \end{align*}
where first we have performed the  integrations over $\y$ and $\y'$ which  leave us the constrain $\ome |\nn| <E'$, and finally we have computed the Fourier transform and used \eqref{agmontrace}.
In the last step we have just rewritten  the energy conservation in a more transparent way.
For an anelastic process, the cross-section is given, see \cite{HF}, by
\beq
\frac{d\sigma_{\nn' \to \nn} }{d \Omega } = \f{|\k|}{|\k'|} \lf| f_{\nn' \to \nn} (\k , \k') \ri|^2.
\eeq
where $ f_{\nn' \to \nn}$ is the scattering amplitude. With our normalization, see \cite{GP, RS}, we have
\beq  
 f_{\nn' \to \nn}(\k , \k') = 2 \pi^2 \T(\k, \nn , \k', \nn').
\eeq
Let us now establish the connection between our strength parameter  $\al$ and the parameter $a$ used by Fermi. 
 For a point interaction in three dimension in the Born approximation,  we have
\beq 
\frac{d\sigma}{d\Omega} = \f{1}{(4\pi \al)^2}.
\eeq
See \cite{Al} for the expression of the scattering amplitude in this case. Comparing with equation (83) of \cite{F}, we obtain that 
\beq 
4 a^2 =  \f{1}{(4\pi \al)^2}  \,. 
\eeq
Notice that on the l.h.s. we have $(2 a)^2$ instead of $a^2$ because in Fermi's analysis $\, 2 a \,$ is the scattering length for the relative motion of the neutron-proton system. 
Thus we find  
\beq \label{quasi}
\frac{d\sigma_{\nn' \to \nn} }{d \Omega}= 4 a^2 \f{|\k|}{|\k'|}  \lf| \int_{\R^3} e^{-i(\k-\k')\cdot \x } \phi_\nn (\x)  \phi_{\nn'} (\x)  d\x \ri|^2
\eeq
which amounts to formula (81) of \cite{F}. In order to arrive at \eqref{fermi}, it is sufficient to follow the arguments used by Fermi.
We sketch them for the reader's sake. We fix $\nn'=0$ and use the integral
\beq
\int_\R e^{ixy} H_n (x) e^{-x^2} dx = \sqrt{\pi} (i \, y)^n e^{-y^2/4}
\eeq
so that \eqref{quasi} reads
\beq \label{daidaidai}
\frac{d\sigma_{{\mathbf 0} \to \nn} }{d \Omega}= 4 a^2 \f{|\k|}{|\k'|} \f{1}{\ome^{|\nn|} } \frac{(\k-\k')_x^{2n_1} }{n_1 !}      \frac{(\k-\k')_y^{2n_2} }{n_2 !}   \frac{(\k-\k')_z^{2n_3} }{n_3 !} e^{ \frac{(\k-\k')^2}{\ome} } .
\eeq
If concentrate on the energy of the final state of the harmonic oscillator and define
\beq
\frac{d\sigma_{ \nn} }{d \Omega} = \sum_{\substack{\mm\in \N_0 \\ |\mm| = |\nn|}} \frac{d\sigma_{{\mathbf 0} \to \mm} }{d \Omega},
\eeq
we obtain 
\beq
\frac{d\sigma_{ \nn} }{d \Omega}=  4 a^2 \f{|\k|}{|\k'|} \f{(\k-\k')^{2|\nn|}}{\ome^{|\nn|}|\nn|! }  e^{ \frac{(\k-\k')^2}{\ome} } 
\eeq
which is \eqref{fermi}. Equations \eqref{fermi2} and \eqref{fermi3} immediately follows from \eqref{fermi}.

\end{document}